\documentclass[11pt]{article}
\usepackage{amsmath,enumitem,amsthm,xy,array,dsfont}
\usepackage{amssymb,bbding,float,caption,multicol,multirow,bm}
\advance\textwidth by 2cm

\advance\textheight by 2.5cm

\usepackage{tikz}
\usepackage{qtree}

\theoremstyle{plain}
\newtheorem{theorem}{Theorem}[section]
\newtheorem{lemma}[theorem]{Lemma}
\newtheorem{cor}[theorem]{Corollary}
\newtheorem*{Assumption*}{Standing assumption}

\theoremstyle{definition}

\newtheorem{example}[theorem]{Example}

\theoremstyle{remark}
\newtheorem{remark}[theorem]{Remark}

\numberwithin{equation}{section}

\newcommand{\timehorizon}{T}
\newcommand{\timeset}{[0,\timehorizon]}

\newcommand{\calF}{\mathcal{F}}
\newcommand{\calP}{\mathcal{P}}
\newcommand{\calA}{\mathcal{A}}

\newcommand{\Y}{\mathcal{Y}}

\newcommand{\calX}{\mathcal{X}}

\newcommand{\abs}[1]{\left\vert#1\right\vert}
\newcommand{\set}[1]{\left\{#1\right\}}

\newcommand\supp{\operatorname{supp}}

\newcommand{\tg}{\tilde{g}}
\newcommand{\hpi}{\hat{\pi}}

\newcommand{\N}{\mathcal{N}}
\newcommand{\eps}{\varepsilon}

\newcommand{\tN}{\widetilde{N}}

\newcommand{\R}{\mathds{R}}

\newcommand{\Uone}{U_1}
\newcommand{\Utwo}{U_2}

\newcommand{\Exp}{\mathds{E}}

\newcommand{\Prob}{\mathbf{P}}

\newcommand{\reals}{\mathds{R}}
\newcommand{\RN}{\mathds{R}}

\newcommand{\jointset}{\Omega \times \timeset}

\newcommand{\tH}{\tilde{H}}

\newcommand{\B}{\mathcal{B}}
\newcommand{\A}{\mathcal{A}}

\newcommand{\Fil}{\mathds{F}}

\makeatletter
\newcommand{\subjclass}[2][1991]{%
  \let\@oldtitle\@title%
  \gdef\@title{\@oldtitle\footnotetext{#1 \emph{Mathematics subject classification.} #2}}%
}
\newcommand{\keywords}[1]{%
  \let\@@oldtitle\@title%
  \gdef\@title{\@@oldtitle\footnotetext{\emph{Keywords.} #1.}}%
}
\makeatother

\usepackage[auth-lg]{authblk}
\title{Utility maximization in jump models driven by marked point processes and nonlinear wealth dynamics}

\author{Mauricio Junca \thanks{E-mail: mj.junca20@uniandes.edu.co}}
\affil{\small Department of Mathematics\\
Universidad de los Andes\\
Cra 1 No. 18a-12\\
Bogot\'a, Colombia}

\author{Rafael Serrano\thanks{E-mail: rafael.serrano@urosario.edu.co} \thanks{Second author gratefully acknowledges the financial support of FIUR research project  DVG170}}

\affil{\small Department of Economics\\ Universidad del Rosario\\
Calle 12C No. 4-69\\
Bogot\'a, Colombia}



\begin{document}
\maketitle





\begin{abstract}
  We explore martingale and convex duality techniques to study optimal investment strategies that maximize expected risk-averse utility from consumption and terminal wealth. We consider a market model with jumps driven by (multivariate) marked point processes and so-called non-linear wealth dynamics which allows to take account of relaxed assumptions such as differential borrowing and lending interest rates or short positions with cash collateral and negative rebate rates. We give sufficient conditions for existence of optimal policies for agents with logarithmic and CRRA power utility and present numerical examples. We provide closed-form solutions for the optimal value function in the case of pure-jump models with jump-size distributions modulated by a two-state Markov chain.
\end{abstract}

\bibliographystyle{amsplain}

\section{Introduction}

The object of this paper is to find sufficient conditions for existence of investment and consumption strategies that maximize expected risk-averse utility from consumption and terminal wealth for an investor trading in a pure-jump incomplete market model that consists of a money market account and a risky asset with price dynamics driven by a (multivariate) marked point process. The investor faces convex trading constraints and additional cash outflow due to market frictions such as differential rates for money borrowing and lending or short positions with negative rebate rates as a result of stock borrowing fees exceeding the interest rate on cash paid as collateral, see Examples \ref{exdiffrates} and \ref{margin} below. This is modelled by adding in a margin payment function that depends on the portfolio proportion process to the investor's wealth equation arising from the self-financing condition.


Marked point processes have gained considerable ground in asset price modelling in the past 15 years, particularly in the modelling of high-frequency financial data and nonlinear filtering for volatility estimation, see e.g. Ceci \cite{ceci2006,ceci2011}, Ceci and Gerardi \cite{cecigerardi,ceci2009}, Cvitani\'{c} et al \cite{cvi2005}, Frey \cite{frey2000}, Frey and Runggaldier \cite{freyrung1999,freyrung2001}, Geman et al \cite{gmy2001}, Prigent \cite{prigent}, Rydberg and Shephard \cite{rydberg}.

Indeed, the random times and marks of the underlying marked point process can be used to model the times of occurrence and the magnitude of different market events such as large trades, limit orders or changes in credit ratings. Market makers update their quotes in reaction to these events, which in turn generates variations and jumps in the stock prices. These jumps can be incorporated in the prices dynamics using the (random) counting measure associated with the underlying marked point process, see the formulation of the asset price model in Section 2.

These processes have also been used for modelling of term structure and forward rates in bond markets, see e.g. Bj\"{o}rk et al \cite{bjorketal} and Jarrow and Madan \cite{jarrowmadan}.




Our approach to the utility maximization problem follows closely the convex duality method started by He and Pearson \cite{hepearson}, Karatzas et al. \cite{karatzas91}, Cvitani\'{c} and Karatzas \cite{cvikar}, and generalized by Kramkov and Schachermayer \cite{kramsch} to the general semi-martingale setting. The method consists in formulating an associated dual minimization problem and finding conditions for absence of duality gap, and has been remarkably effective in dealing with utility maximization with convex portfolio constraints.

For jump-diffusion market models, Goll and Kallsen \cite{goll2000}, Kallsen \cite{kallsen2000} and more recently Michelbrink and Le \cite{michelbrink} use the martingale approach to obtain explicit solutions for agents with logarithmic and power utility functions and linear wealth equations. Callegaro and Vargiolu \cite{callevar} obtain similar results in jump-diffusion models with Poisson-type jumps. In the diffusion setting, the convex duality approach was significantly extended by Cuoco and Liu \cite{cuocoliu} and Klein and Rogers \cite{kleinrogers} in order to incorporate non-linear wealth dynamics.

In this paper, we consider the same utility maximization as in Cuoco and Liu \cite{cuocoliu} but with a marked point process, instead of a Brownian motion, as the main driving process. Our main assumption throughout is that the counting measure $N(dy,dt)$ of the underlying marked point process has local characteristics of the form $(\lambda_t,F_t(dy))$. Of particular interest is the case in which these local characteristics depend on an (possibly exogenous) Markovian state process which may describe intra-day trading activity, macroeconomics factors, microstructure rules that drive the market, or simply changes in the economy or business cycle, see Examples \ref{MMjumps} and \ref{exceci} below. See also Frey and Runggaldier \cite{freyrung2001}.

The main result of this paper is a sufficient condition for existence of an optimal portfolio-consumption pair in terms of the convex dual of the margin payment function and the solution pair of a linear backward SDE driven by the counting measure $N(dy,dt).$ Although the optimality condition in the main result seems rather restrictive, in the last section we show that it simplifies significantly in the case of logarithmic and power utility functions as well as premium payments due to higher borrowing interest rates or short selling.

As our main example, we consider the regime-switching pure-jump asset price model proposed by L\'opez and Ratanov \cite{lopezrat} in which the jump-size distributions alternate according to a continuous-time two-state Markov chain, see also the recent papers by Elliot and Siu \cite{elliottsiu} and L\'opez and Serrano \cite{lopezserrano}. We find explicit formulae for the optimal value functions for this model in the case of logarithmic utility.


Finally, it is worthwhile mentioning that our market model and formulation of the utility maximization problem can be seen as particular case of the more general problem studied by Schroder and Skiadas \cite{schroder2008} as they consider a jump-diffusion model driven by a Brownian motion and a marked point process as well as recursive utility functions. However, their main approach is the scale/translation-invariant formulation of the utility maximization problem, 
although they do relate their results to the dual formulation in the appendix of \cite{schroder2008}.

Let us briefly describe the contents of this paper. In Section 1 we outline the stochastic setting and information structure for marked point processes, introduce the market model and non-linear wealth equation with margin payment functions and define the optimal investment/consumption problem. In Section 3 we formulate the main assumption on existence of local characteristics for the underlying marked point process. In Section 4 we introduce convex duality techniques from Cuoco and Liu \cite{cuocoliu} and establish our main result on a sufficient condition for existence of an optimal investment/consumption policy. In Section 5 we illustrate the main result by considering (CRRA) logarithmic and power utility functions combined with margin payments for differential borrowing and lending interest rates as well as short positions with negative rebate rates. We present some numerical examples and, using recent results in L\'opez and Serrano \cite{lopezserrano}, provide explicit closed-form solutions for the optimal value function in the case of Markov-modulated jump-size distributions.






\section{Market model, non-linear wealth dynamics and risk-averse utility maximization problem}

Let $(\Omega,\Prob,\calF)$ be a complete probability space endowed with  a filtration $\Fil=\set{\calF_t}_{t\geq 0}$ and let $E$ be Borel subset of an Euclidean space with $\sigma-$algebra $\mathcal{B}(E).$ Let $\set{(\tau_n,Y_n)}_{n\geq 1}$ be a \emph{marked (or multivariate) point process} with mark space $E$, that is, $\set{Y_n}_{n\geq 1}$ is a sequence of $E-$valued random variables and $\set{\tau_n}_{n\geq 1}$ is an increasing sequence of positive random variables satisfying $\lim_{n\to\infty}\tau_n=+\infty.$

We define the random counting measure $N(dy,dt)$ associated with the marked point process $\set{(\tau_n,Y_n)}_{n\geq 1}$ as follows
\begin{equation}\label{countN}
N(A\times(0,t]):=\sum_{n=1}^\infty\mathbf{1}_{\set{\tau_n\le t,Y_n\in A}}=\sum_{\tau_n\le t}\mathbf{1}_{\set{Y_n\in A}}, \ \ A\in\mathcal{B}(E), \ \ t\geq 0
\end{equation}
see e.g. Jacod and Shiryaev \cite[Chapter III, Definition 1.23]{js} or Jeanblanc et al \cite[Section 8.8]{jeanblanc}. For each $A\in\B(E),$ the counting process $N_t(A):=N(A\times(0,t])$ counts the number of marks with values in $A$ up to time $t.$

Let $\Fil^N=\set{\calF^N_t}_{t\geq 0}$ denote the natural filtration related to these counting processes, that is
\[
\calF^N_t:=\sigma(N_s(A):0\le s\le t, \ A\in\B(E)), \ \ t\geq 0.
\]
Throughout we assume $\Fil^N\subset\Fil.$ Recall that a real-valued process $(\phi_t)_{t\geq 0}$ is $\mathds{F}-$predictable if the random function $\phi(t,\omega)=\phi_t(\omega)$ is measurable with respect to the $\sigma-$algebra $\calP$ on $\R_+\times\Omega$ generated by adapted left-continuous processes.
Similarly, a map $\phi:\Omega\times\R_+\times E\to\R$ is said to be a $\Fil$-predictable if it is measurable with respect to the product $\sigma$-algebra $\calP\otimes\B(E).$

We consider a financial market model with a money account $B_t$ with continuously compounded force of interest $r_t$
\[
B_t=\exp\left(\int_0^t r_s\,ds\right), \ \ t\geq 0,
\]
and a risky asset or stock with price process defined as the stochastic exponential process $S_t:=S_0\mathcal{E}_t(L)$ with $S_0>0$ and
\[
L_t:=\int_0^t\mu_s\,ds+\sum_{\tau_n\le t}f(\tau_n,Y_n), \ \ t\geq 0.
\]
The processes $r_t,\mu_t$ and the map $f(t,y)$ are assumed uniformly bounded and  $\Fil$-predictable. Here $\mathcal{E}_t(\cdot)$ denotes the stochastic (Dol\'{e}ans-Dade) exponential, see e.g. Jeanblanc et al \cite[Section 9.4.3]{jeanblanc}.

In this model, the discrete random times $\tau_n$ can be interpreted as time points at which significant market events occur such as large trades, limit orders or changes in credit ratings, or simply times at which market makers update their quotes in reaction to new information. The marks $Y_n$ describe the magnitude of these events, and both times $\tau_n$ and marks $Y_n$ create jumps and variations in the stock prices through the map $f(t,y).$

Notice that $B_t$ and $S_t$ satisfy
\begin{align}
{dB_t}&=B_tr_{t}\,dt, \ \ B_0=1\notag\\
{dS_t}&=S_{t-}\Bigl[\mu_{t}\,dt+\int_E f(t,y)\,N(dy,dt)\Bigr], \ \ S_0>0.\label{eqS}
\end{align}
We also assume $f(t,y)>-1$ a.s. for every $(t,y)\in[0,T]\times E.$ Thus, the solution to equation (\ref{eqS}) is given by the predictable process
\begin{align*}
S_t&=S_0\exp\left(\int_0^t\mu_s\,ds\right)\prod_{\tau_n\le t}(1+f(\tau_n,Y_n))\\
&=S_0\exp\biggl(\int_0^t\mu_s\,ds+\sum_{n=1}^{N_t(E)}\ln(1+f(\tau_n,Y_n))\biggr)\\
&=S_0\exp\left(\int_0^t\mu_s\,ds+\int_0^t\int_{E}\ln(1+f(s,y))\,N(dy,ds)\right), \ \ t\in [0,T].
\end{align*}


For an agent willing to invest in this market, we denote with $\pi_t$ the fraction of wealth invested in the risky asset, so that the fraction of wealth invested in the riskless asset is $1-\pi_t.$ Recall that a positive value for $\pi_t$ represents a long position in the risky asset, whereas a negative $\pi_t$ stands for a short position. The process $\pi=(\pi_t)_{t\in [0,T]}$ is called \emph{portfolio proportion process}, or simply portfolio process, and we always assume it is $\Fil$-predictable.

We fix an finite investment horizon $T>0$  and a non-empty closed convex $K\subset\R$ of {\em portfolio constraints} with $0 \in K.$  We introduce a {\em margin payment function} $g:\jointset\times\RN\rightarrow\reals$ which is $\calP\times\B(\RN)$-measurable and satisfies, for each $(\omega,t)\in\Omega\times[0,T]$
\begin{enumerate}
  \item[(i)] $g(\omega,t,0) = 0,$
  \item[(ii)] $g(\omega,t,\cdot)$ is concave and continuous on $\RN.$
\end{enumerate}
During the time interval $[0,T],$ the investor is allowed to consume at an instantaneous consumption rate $c_t\geq 0.$ Then, under the  self-financing condition, the wealth $V_t^{x,\pi,c}$ of the investor at time $t\in[0,T]$ is subject to the following dynamic budget constraint
\begin{equation}\label{eqVnon-linear}
\begin{split}
dV_t&=V_{t-}\left\{\left[r_t+g(t,\pi_t)\right]\,dt+\pi_t\Bigl[(\mu_t-r_t)\,dt+\int_E f(t,y)\,N(dy,dt)\Bigr]\right\}-c_t\,dt,\\
V_0&=x,
\end{split}
\end{equation}
where $x>0$ is a fixed initial wealth and $(\pi,c)=(\pi_t,c_t)_{t\in [0,T]}$ is a pair of $\Fil$-predictable portfolio-consumption processes satisfying
\begin{enumerate}
  \item[(i)] $\pi$ is bounded below and $\pi_t\in K$ a.s. for all $t\in[0,T],$

  \smallskip

  \item[(ii)] $\displaystyle{\int_0^T [\pi_t^2+g(t,\pi_t)+c_t]\,dt<+\infty,}$ a.s.
\end{enumerate}
The stochastic differential equation (\ref{eqVnon-linear}) is linear with respect to the wealth process $V^{x,\pi,c}$ but possibly non-linear in the portfolio policy $\pi$ and, in turn, with respect to the actual amounts invested in both the risky and non-risky asset.

We define the class $\calA(x)$ of admissible pairs for initial wealth $x>0$ as the set of portfolio-consumption pairs $(\pi,c)$  for which wealth equation (\ref{eqVnon-linear}) possesses an unique strong solution satisfying $V_t\geq 0,$ a.s. for all $t\in[0,T].$

We denote with $V^{x,\pi,c}=(V_t^{x,\pi,c})_{t\in [0,T]}$  the solution to equation (\ref{eqVnon-linear}). In particular, if there is no consumption i.e. $c_t =0$ for all $t\in[0,T],$ equation (\ref{eqVnon-linear}) is linear in $V_t$ and homogeneous, with solution
\begin{align}
V_t^{x,\pi,0}&=x\mathcal{E}_t\left(\int_0^\cdot \left[r_s+g(s,\pi_s)+\pi_s\left(\mu_s-r_s\right)\right]\,ds+\int_0^\cdot\int_E\pi_s f(s,y)\,N(dy,ds)\right)\notag\\
&=x\exp\left(\int_0^t\left[r_s+g(s,\pi_s)+\pi_s\left(\mu_s-r_s\right)\right]\,ds\right)\prod_{n=1}^{N_t(E)}(1+\pi_{\tau_n}f(\tau_n,Y_n)).\label{Vpi0}
\end{align}
This is always positive if, for instance, short-selling is not allowed for any of the assets i.e. if $\pi_t\in [0,1]$ for all $t\in[0,T].$
We can use (\ref{Vpi0}) to find an expression for the wealth process $V^{x,\pi,c}$ in terms of $V^{1,\pi,0}=(V_t^{1,\pi,0})_{t\in [0,T]},$ the wealth process with initial wealth $1$ and portfolio-consumption pair $(\pi,0),$ as follows: define the process
\[
\xi_t^{x,\pi,c}:=x-\int_0^t\frac{c_s}{V_{s-}^{1,\pi,0}}\,ds, \ \ t\in [0,T].
\]
In differential form, we have $V_{t-}^{1,\pi,0}\,d\xi_t^{{x,\pi,c}}=-c_t\,dt.$ Then
\begin{align*}
d&\left(\xi_t^{{x,\pi,c}}V_t^{1,\pi,0}\right)\\
&=\xi_t^{{x,\pi,c}}\,dV_t^{1,\pi,0}+V_{t-}^{1\pi,0}\,d\xi_t^{{x,\pi,c}}\\
&=\xi_t^{{x,\pi,c}}V_{t-}^{1,\pi,0}\left\{\left[r_t+g(t,\pi_t)\right]\,dt+\pi_t\Bigl[(\mu_t-r_t)\,dt+\pi_t\int_{E} f(t,y)\,N(dy,dt)\Bigr]\right\}-c_t\,dt.
\end{align*}
Since $\xi_0^{{x,\pi,c}}V_0^{1,\pi,0}=x,$ by uniqueness of solution to equation (\ref{eqVnon-linear}), the wealth process $V^{x,\pi,c}$ is a modification of the process
\begin{equation}\label{xiV}
  \begin{split}
&\xi_t^{{x,\pi,c}}V_t^{1,\pi,0}\\
&\phantom{\gamma t}=\biggl[x-\int_0^t\frac{c_s}{V_{s-}^{1,\pi,0}}\,ds\biggr]
\exp\left(\int_0^t\left[r_s+g(s,\pi_s)+\pi_s\left(\mu_s-r_s\right)\right]\,ds\right)\prod_{n=1}^{N_t(E)}(1+\pi_{\tau_n}f(\tau_n,Y_n)).
\end{split}
\end{equation}
Notice that the portfolio-consumption pair $(\pi,c)$ leads to positive wealth at time $t\in [0,T]$ if, almost surely
\[
\int_0^t\frac{c_s}{V_{s-}^{1,\pi,0}}\,ds<x \ \ \mbox{ and } \ \ \pi_{\tau_n}f(\tau_n,Y_n)>-1, \ \forall \tau_n\le t.
\]

\begin{example}[Different borrowing and lending rates]\label{exdiffrates}
Consider the case of a financial market in which the interest rate that an investor pays for borrowing is higher than the bank rate the investor earns for lending.

More concretely, let the borrowing rate $R_t$ be a $\Fil$-predictable uniformly bounded process satisfying $R_t\geq r_t$ a.s. for all $t\in [0,T].$ The margin payment function for this case is
\[
g(t,\pi):=-(R_t-r_t)(\pi-1)^+, \ \ \pi\in\R.
\]
Notice that the wealth equation (\ref{eqVnon-linear}) is nonlinear in the portfolio process, although it is piecewise linear. The term in the equation that involves the portfolio $\pi_t$ and the interest rates reads
\[
r_t+g(t,\pi_t)+\pi_t(\mu_t-r_t)=
\left\{
  \begin{array}{ll}
    \pi_t\mu_t+(1-\pi_t)r_t, & \hbox{if} \ \pi_t\le 1, \\\\
    \pi_t\mu_t+(1-\pi_t)R_t, & \hbox{if} \ \pi_t\geq 1.
  \end{array}
\right.
\]
\end{example}

\begin{example}[Short selling with cash collateral and negative rebate rates]\label{margin}
In a short sale transaction, an investor borrows stock shares --typically from a broker-dealer or an institutional investor-- and sells them on the market, and at some point in the future must buy the shares back to return them to the lender in the hope of making a profit if the share price decreases.

In order to borrow the stock shares, the short-seller must engage in a security-lending agreement and pay a loan fee to the lender. This fee depends mostly on the difficulty to locate and borrow the security. The short-seller must also put up collateral to better insure that the borrowed shares will be returned to the lender. Acceptable collateral includes cash, government bonds or a letter of credit from a bank. If the collateral is cash, the lender ``rebates" interest on the collateral to the borrower, which usually offsets the stock loan fee.

However, if there is a large demand for the security, the stock loan fee  might exceed the cash collateral interest rate, and the rebate rate will be negative. We assume this is the case for the risky asset $S_t.$ Moreover, we assume that the cash proceeds from the short-sale are held as collateral and earn interests at the money account rate $r_t.$

More concretely, let the stock loan fee $r_t^L$ be a $\Fil$-predictable uniformly bounded process satisfying $r_t^L\geq r_t$ a.s. for all $t\in [0,T].$ Then, the margin payment function is
\[
g(t,\pi)=(r_t-r_t^L)\pi^-.
\]
Here $\pi^{-}:=-\min\set{0,\pi}$ denotes the negative part of $\pi\in\R.$ We refer the reader to the paper by D'Avolio \cite{davolio2002} for a comprehensive description of the market for borrowing and lending stocks, short selling as well as empirical evidence on negative rebate rates.




\end{example}

We now introduce the risk-averse utility maximization problem for optimal choice of portfolio and consumption processes. Let $\Uone : [0,T] \times [0,\infty) \rightarrow [-\infty,\infty)$ and $\Utwo : [0,\infty) \rightarrow [-\infty,\infty)$ denote consumption and
investment risk-averse utility functions respectively, satisfying the following conditions
%
\begin{enumerate}
\item[(i)] $\Uone(t,x) > -\infty$
and $\Utwo(x) > -\infty$
for all $t\in \timeset$
and $x \in (0,\infty),$
%
%

\item[(ii)] for each $t \in [0,T]$ the mappings
$\Uone(t,\cdot): (0,\infty) \rightarrow \reals$
and $\Utwo(\cdot) : (0,\infty) \rightarrow \reals$
are strictly increasing, strictly concave, of class $C^{1}$ on $(0,\infty),$
such that
\[
\lim_{x \downarrow 0, \; x > 0} \frac{\partial \Uone}{\partial x}(t,x) = +\infty,\;\;\;
\lim_{x \rightarrow \infty}  \frac{\partial \Uone}{\partial x}(t,x) = 0,\;\;\;
\lim_{x \downarrow 0, \; x > 0} \Utwo'(x) = +\infty,\;\;\;
\lim_{x \rightarrow \infty}  \Utwo'(x) = 0.
\]
\item[(iii)] $\Uone$ and
$\frac{\partial \Uone}{\partial x}$
are continuous on $[0,T]\times (0,\infty)$.
\end{enumerate}%

%

Let $\tilde{\calA}(x)$ denote the class of admissible portfolio-consumption strategies $(\pi,c)\in\calA(x)$ such that
\[
\Exp\left[\int_0^T U_1(t,c_t)^{-}\,dt+U_2(V_T^{x,\pi,c})^{-}\right]<+\infty.
\]
We define the utility functional
\[
J(x;\pi,c):=\Exp\left[\int_0^T U_1(t,c_t)\,dt+U_2(V_T^{x,\pi,c})\right], \ \ (\pi,c)\in\tilde{\calA}(x).
\]
Our main object of study is the following risk-averse utility maximization problem
\begin{equation}\label{utilitymax}
\vartheta(x):=\sup_{(\pi,c)\in\tilde{\calA}(x)}J(x;\pi,c), \ \ x>0.
\end{equation}

\section{Main assumption: local characteristics}
Let $N(dy,dt)$ denote the random counting measure (\ref{countN}). Recall that the compensator or $\Fil$-predictable projection $\rho(dy,dt)$ of  $N(dy,dt)$ is the unique (possibly, up to a null set) positive random measure such that, for every $\Fil$-predictable real-valued map $\phi(t,y)$ the two following conditions hold
\begin{itemize}
  \item[i.] The process
  \[
  \int_0^t\int_E\phi(s,y)\,\rho(dy,ds), \ \ t\geq 0,
  \]
  is $\Fil$-predictable.

  \smallskip\item[ii.] If the process
  \[
  \int_0^t\int_E\abs{\phi(s,y)}\,\rho(dy,ds)<+\infty, \ \ \forall t\geq 0.
\]
is increasing and locally integrable, then
\[
M_t(\phi):=\int_0^t\int_E \phi(s,y)\,N(dy,ds)-\int_0^t\int_E\phi(s,y)\,\rho(dy,ds), \ \ t\geq 0,
\]
is $\Fil$-local martingale (see e.g. Jeanblanc et al \cite[Definition 8.8.2.1]{jeanblanc}). Equivalently, for all $T>0,$
\[
\Exp\left[\int_0^T\int_E\phi(s,y)\,N(dy,ds)\right]=
\Exp\left[\int_0^T\int_E\phi(s,y)\,\rho(dy,ds)\right].
\]
\end{itemize}
The following is the main standing assumption for the rest of this paper: the compensator $\rho(dy,dt)$ of the counting measure $N(dy,dt)$ satisfies
\begin{equation}\label{stand}
\rho(dy,dt)=F_t(dy)\,\lambda_t\,dt \ \tag{A}
\end{equation}
where $\lambda_t$ is a positive $\Fil$-predictable process and $F_t(dy)$ is a predictable probability transition kernel, that is, a $\Fil$-predictable process with values in the set probability measures on $\B(E).$ In this case $(\lambda_t,F_t(dy))$ are called the $\Fil$-local characteristics of the marked point process $\set{(\tau_n,Y_n)}_{n\geq 1},$ see e.g. Br\'{e}maud \cite[Chapter VIII]{bremaud}.

Under this assumption, for each $A\in\B(E)$ the counting process $N_t(A)$ is an inhomogeneous Poisson process with stochastic intensity $\lambda_t F_t(A).$ This can be interpreted as it is possible to separate the probability that an event occurs from the conditional distribution of the mark, given that the event has occurred. Thus, $F_t(dy)$ is the conditional distribution of the mark at time $t,$ and $\lambda_t\,dt$ gives the probability of an event occurring in the next infinitesimal time step $dt.$


Below we present two examples that satisfy the main assumption (\ref{stand}). In both cases, the $\Fil$-local characteristics $(\lambda_t,F_t(dy))$ depend on an (possibly exogenous) Markovian state process with RCLL paths which may be used to describe intra-day market activity, macroeconomics factors, microstructure rules that drive the market or changes in the the economy or business cycle, see e.g \cite{freyrung2001}. In the first example, the state process is a two-state continuous time Markov-chain. In the second example, it is a jump-diffusion process, possibly having common jumps with the risky asset $S_t.$
\begin{example}[Markov-modulated marked point process]\label{MMjumps}
Let $\{\eps(t)\}_{t\geq0}$ be a two-state continuous-time Markov chain with values in $\{0,1\}$ and infinitesimal generator (intensity matrix)
\begin{equation*}
Q=
\begin{pmatrix}
-\bm{\lambda}_0 & \ \ \bm{\lambda}_0\\
\bm{\ \ \lambda}_1 & -\bm{\lambda}_1
\end{pmatrix}.
\end{equation*}
For each $i=0,1$, let $\{Y_{i,n}\}_{n\geq 1}$ be a sequence of independent $E-$valued random variables with distributions
\[
\Prob(Y_{i,n}\in dy)=\bm{F}_i(dy), \ \ n\geq 1.
\]
Suppose the two distributions $\bm{F}_0$ and $\bm{F}_1$ are independent as well as independent of the Markov chain $\{\eps(t)\}_{t\geq0}.$ Let $\set{\tau_n}_{n\geq 1}$ denote the jump times of $\{\eps(t)\}_{t\geq0}$ and let $\eps_n:=\eps(\tau_n-)$ denote the state of the Markov chain $\{\eps(t)\}_{t\geq0}$ right before the $n-$th jump.

Thus, for each $n\geq 1$ the mark $Y_{\eps_n,n}$ is a random variable with distribution $\bm{F}_{\eps_n}(dy).$ Then, the $\Fil-$predictable projection $\rho(dy,dt)$ of counting measure $N(dy,dt)$ related to the marked point process $\set{(\tau_n,Y_{\eps_n,n})}_{n\geq 1}$  satisfies condition (\ref{stand}) with $\Fil$-local characteristics
\[
\lambda_t=\bm{\lambda}_{\eps(t-)} \ \ \mbox{and} \ \ F_t=\bm{F}_{\eps(t-)}.
\]
For the proof see Section 4 in L\'opez and Serrano \cite{lopezserrano}.
\end{example}

\begin{example}\label{exceci}
Let $\nu(d\xi)$ be a $\sigma-$finite measure on a measurable space $\mathcal{Z}.$ Let $\gamma(d\xi,dt)$ denote a $\Fil$-Poisson random measure on $\mathcal{Z}\times\R_+$ with mean measure $\nu(d\xi)\,dt,$ and let $W_t$ be a $\Fil$-Brownian motion.

Let $(X,Z)$ denote the solution of the system of It\^{o}-Levy type stochastic differential equations
\begin{align*}
dX_t&=b(X_t)\,dt+\sigma(X_t)\,dW_t+\int_\mathcal{Z}\eta_0(t,X_{t-},\xi)\,\gamma(d\xi,dt), \ \ X_0\in\R\\
  dZ_t&=\int_\mathcal{Z}\eta_1(t,X_{t-},Z_{t-},\xi)\,\gamma(d\xi,dt), \ \ Z_0\in\R.
\end{align*}
We assume the Poisson measure $\gamma(d\xi,dt)$ is independent of $W_t$ and the coefficients $b,\sigma$ and $\eta_0$ are measurable functions of their arguments satisfying the usual linear growth and Lipschitz conditions.

Define the sequence of random times $\set{\tau_n}_{n\geq 1}$ as the jump times of the process $Z_t,$ that is
\begin{align*}
  \tau_0&:=0\\
\tau_{n+1}&:=\inf\set{t>\tau_n:\int_{\tau_n}^t\int_\mathcal{Z}\eta_1(s,X_{s-},Z_{s-},\xi)\,\gamma(d\xi,ds)\neq 0}, \ \ n=1,2,\ldots
\end{align*}
For each $n\geq 1,$ the mark $Y_n$ (with mark space $E=\R)$ is the jump of the process $Z_t$ at time $\tau_n,$
\[
Y_n:=\Delta Z_{\tau_n}=Z_{\tau_n}-Z_{\tau_{n-1}}.
\]
For $t\geq 0$ and $A\in\mathcal{B}(\R)$ we define the sets
\[
D_t(x,z;A):=[\eta_1(t,x,z,\cdot)^{-1}](A\setminus\set{0})=\set{\xi\in\mathcal{Z}:\eta_1(t,x,z,\xi)\in A\setminus\set{0}}.
\]
Then, if
\[
\int_0^T\nu(D_t(X_t,Z_t;\R))\,dt <+\infty, \ \ \Prob-\mbox{a.s.}
\]
the $\Fil-$predictable projection $\rho(dy,dt)$ of the counting measure $N(dy,dt)$ associated with $\set{(\tau_n,Y_n)}_{n\geq 1}$ satisfies condition (\ref{stand}) with $\Fil-$local characteristics
\begin{align*}
\lambda_t&=\nu(D_t(X_{t-},Z_{t-};\R))\\
F_t(dy)&=\frac{1}{\lambda_t}\nu(D_t(X_{t-},Z_{t-};dy)).
\end{align*}
For the proof see Proposition 2.2 in Ceci \cite{ceci2006}. 
\end{example}

\section{Convex duality approach and main result}
In this section we introduce some of the convex duality techniques from Cuoco and Liu \cite{cuocoliu} and establish our main result on a sufficient condition for existence of an optimal investment/consumption policy. Let
\[
\delta_K(\pi):=
\left\{
  \begin{array}{ll}
    0, & \hbox{if} \ \pi\in K \\
    +\infty, & \hbox{if}  \ \pi\notin K
  \end{array}
\right.
\]
denote the indicator function (in the sense of convex analysis) of the portfolio constraint set $K,$  and let
\[
g_K(\omega,t,\pi):=g(\omega,t,\pi)-\delta_K(\pi).
\]
The function $g_K(\omega,t,\cdot)$ is upper semi-continuous and concave with $g_K(\omega,t,0)=0$ a.s., for all $t\in[0,T].$ We denote by
\[
\tg_K(\omega,t,\zeta):=\sup_{\pi\in\R}[g_K(\omega,t,-\pi)+\pi\zeta]=\sup_{\pi\in K}[g(\omega,t,\pi)-\pi\zeta], \ \ \zeta\in\R
\]
the convex conjugate of $\R\ni\pi\mapsto -g_K(t,-\pi)\in\R.$ Since $0\in K,$ it is clear from the definition that $\tg_K\geq 0.$ Moreover, $\tg_K(\omega,t\,\cdot)$ is lower semi-continuous, convex and
\begin{equation}\label{gk}
g_K(\omega,t,\pi)=\inf_{\zeta\in\R}[\tg_K(\omega,t,\nu)+\pi\zeta].
\end{equation}
If $K=\R$ we denote $\tg_K$ simply with $\tg.$ We define the effective domain of $\tg_K(\omega,t,\zeta),$ denoted with $\N_t,$ as
\[
\N_t(\omega):=\set{\zeta\in\R:\tg_K(\omega,t,\zeta)<+\infty}.
\]
Finally, let $\mathcal{D}$ denote the set of $\Fil$-progressively measurable processes $(\zeta_t)_{t\in [0,T]}$ satisfying
\[
\sup_{t\in [0,T]}\abs{\zeta_t}+\int_0^T\tg_K(t,\zeta_t)\,dt<+\infty, \ \ \mbox{a.s.}
\]

\begin{example}\label{exdiffrates2}
Consider the margin payment function of Example \ref{exdiffrates}, under the portfolio constraint of prohibition of short-selling of the risky asset, that is $K=[0,+\infty).$ For $\zeta\in\R$ fixed, the map
\[
\R\ni\pi\mapsto g(t,\pi)-\pi\zeta=
\left\{
  \begin{array}{ll}
    -\pi\zeta, & \ \pi\in [0,1],\\
    -(\zeta+R_t-r_t)\pi+(R_t-r_t), &  \ \pi > 1,
  \end{array}
\right.
\]
attains a finite maximum value if and only if $-(\zeta+R_t-r_t)\le 0.$ This maximum value is attained at $\pi=1$ if $-\zeta\ge 0$ and at $\pi=0$ if $-\zeta\le 0.$ Hence, we have
\begin{equation}\label{gkdiffrates}
\tg_K(t,\zeta)=\left\{
  \begin{array}{ll}
  0, &  \ \zeta\geq 0, \\
     -\zeta, &  \ \zeta\in [-(R_t-r_t),0], \\
    +\infty, &  \ \zeta <-(R_t-r_t)
  \end{array}
\right.
\end{equation}
with effective domain $\N_t= [-(R_t-r_t),+\infty).$
\end{example}

\begin{example}\label{margin2}
Consider now the margin payment function of Example \ref{margin}, with portfolio constraint of prohibition of borrowing from the money account, that is $K=(-\infty,1].$  For $\zeta\in\R$ fixed, the map
\[
\R\ni\pi\mapsto g(t,\pi)-\pi\zeta=
\left\{
  \begin{array}{ll}
    -\pi\zeta, & \ \pi\in [0,1],\\
   (r_t^L-r_t-\zeta)\pi, &  \ \pi < 0 ,
  \end{array}
\right.
\]
attains a (finite) maximum value if and only if $r_t^L-r_t-\zeta\ge 0.$ Again, this maximum value is attained at $\pi=1$ if $-\zeta\ge 0$ and at $\pi=0$ if $-\zeta\le 0.$  Then, we have
\begin{equation}\label{gkmargin}
\tg_K(t,\zeta)=\left\{
  \begin{array}{ll}
    0, & \ \zeta\le 0\\

    -\zeta, &  \  \zeta\in [0,r_t^L-r_t]\\

    +\infty, &  \  \zeta> r_t^L-r_t

  \end{array}
\right.
\end{equation}
with effective domain $\N_t=(-\infty,r_t^L-r_t].$
\end{example}




Let $\Theta$ denote the set of locally bounded $\Fil$-predictable $E-$marked processes $\varphi(t,y)$ satisfying
\begin{enumerate}
  \item[(i)] $\varphi(t,y)>0$ a.s. for $\rho-$almost every $(t,y)\in[0,T]\times E$

  \item[(ii)] The process $(\zeta^\varphi_t)_{t\in [0,T]}$ defined as
  \[
  \zeta^\varphi_t:=r_t-\mu_t-\lambda_t\int_Ef(t,y)\varphi(t,y)\,F_t(dy), \ \ t\in [0,T]
  \]
  belongs to $\mathcal{D}.$
\end{enumerate}
Let $\widetilde{N}(dy,dt):=N(dy,dt)-F_t(dy)\lambda_t\,dt$ denote the compensated martingale measure of the counting measure $N(dy,dt).$ For each $\varphi\in\Theta$ let $H^{\varphi}$ denote the solution of the linear SDE
\begin{equation}\label{eqH}
\begin{split}
  dH_t&=H_{t-}\left\{-[r_t+\tg_K(t,\zeta^\varphi_t)]\,dt+\int_E(\varphi(t,y)-1)\,\tN(dy,dt)\right\}\\
H_0&=1
\end{split}
\end{equation}
\begin{lemma}
For each $\varphi\in\Theta$ and $(\pi,c)\in\A(x),$ we have
\begin{equation}\label{budget}
\Exp\left[H_T^{\varphi} V_T^{x,\pi,c}+\int_0^TH_T^{\varphi}c_s\,ds\right]\le x.
\end{equation}
\end{lemma}
\begin{proof}
Using the product rule for jump processes, we obtain
\begin{align*}
  d&(H_t^\varphi V_t^{x,\pi,c})+H_t^\varphi c_t\,dt\\
  &=H_{t-}^\varphi dV_t^{x,\pi,c}+V_{t-}^{x,\pi,c}dH_{t}^\varphi +H_{t-}^\varphi V_{t-}^{x,\pi,c}\pi_t\int_Ef(t,y)(\varphi(t,y)-1)\,N(dy,dt)+H_t^\varphi c_t\,dt\\
  &=H_{t-}^\varphi V_{t-}^{x,\pi,c}\left\{\left[r_t+g(t,\pi_t)\right]\,dt+\pi_t\Bigl[(\mu_t-r_t)\,dt+\int_E f(t,y)\,N(dy,dt)\Bigr]\right\}-H_t^\varphi c_t\,dt\\
  &\phantom{=}+H_{t-}^\varphi V_{t-}^{x,\pi,c}\left\{-[r_t+\tg_K(t,\zeta^\varphi_t)]\,dt+\int_E(\varphi(t,y)-1)\,\tN(dy,dt)\right\}\\
  &\phantom{=}+H_{t-}^\varphi V_{t-}^{x,\pi,c}\pi_t\int_Ef(t,y)(\varphi(t,y)-1)\,N(dy,dt)+H_t^\varphi c_t\,dt\\
  &=H_{t-}^\varphi V_{t-}^{x,\pi,c}\left\{\Bigl[g(t,\pi_t)-\pi_t\Bigl(r_t-\mu_t-\lambda_t\int_Ef(t,y)\varphi(t,y)\,F_t(dy)\Bigr)
  -\tg_K(t,\zeta^\varphi_t)\Bigr]\,dt\right.\\
  &\phantom{=}+\left.\int_E(\pi_tf(t,y)\varphi(t,y)+\varphi(t,y)-1)\,\tN(dy,dt)\right\}
\end{align*}
Integrating from $0$ to $T,$ and using the definition of $\tg_K,$ we get
\begin{equation*}
  H_T^\varphi V_T^{x,\pi,c}+\int_0^T H^\varphi_t c_t\,dt
  \le x+ \int_0^T\int_E(\pi_tf(t,y)\varphi(t,y)+\varphi(t,y)-1)\,\tN(dy,dt), \ \mbox{a.s}
\end{equation*}
The stochastic integral in the right hand side of the last inequality is a $\Fil$-local martingale which is bounded below, hence a $\Fil-$super martingale, and (\ref{budget}) follows.
\end{proof}
We now introduce an auxiliary functional related to the convex dual of the utility functions. Let $U$ denote either $U_2(\cdot)$ or $U_1(t,\cdot)$ with $t\in[0,T]$ fixed. Let $I$  denote the inverse of $U',$  so that
\[
I(U'(x))=U'(I(x))=x, \ \ \forall x>0.
\]
Then, $I$ satisfies
\[
I(y)=\arg\max_{x>0}\set{U(x)-yx}, \ \ y>0.
\]
In particular,
\begin{equation}\label{ineqUI}
U(I(y))-yI(y)\geq U(x)-yx, \ \ \forall x,y>0.
\end{equation}
Notice that $U(I(y))-yI(y)=U^*(y),$ where $U^*(y):=\sup_{x>0}\set{U(x)-yx}$ is the Legendre-Fenchel transform of the map $(-\infty,0)\ni x\mapsto -U(-x)\in\R.$ The map $U^*$ is known as the convex dual of the utility function $U.$

For each $\varphi\in\Theta,$ we define the map
\[
\calX^{\varphi}(y):=\Exp\left[\int_0^T H^{\varphi}_tI_1(t,yH_t^{\varphi})\,dt+H^{\varphi}_TI_2(yH^{\varphi}_T)\right].
\]
Let $\widetilde{\Theta}:=\set{\varphi\in\Theta:\calX^{\varphi}(y)<\infty, \ \forall y>0}.$ For each $\varphi\in\widetilde{\Theta}$ we denote $\Y^{\varphi}:=(\calX^{\varphi})^{-1}$ and define the process $(c_t^{x,\varphi})_{t\in [0,T]}$ and random variable $G^{x,\varphi}$ as follows
\begin{equation}\label{cY}
\begin{split}
  c_t^{x,\varphi} &:=I_1(t,\Y^{\varphi}(x)H_t^{\varphi}), \ t\in [0,T],\\
  G^{x,\varphi}&:=I_2(\Y^{\varphi}(x)H_T^{\varphi}).
\end{split}
\end{equation}
Finally, we define the auxiliary functional
\[
L(x;\varphi):=\Exp\left[\int_0^T U_1(t,c_t^{x,\varphi})\,dt+U_2(G^{x,\varphi})\right], \ \ x>0, \ \varphi\in\widetilde{\Theta}.
\]

\begin{lemma}\label{ineqJL}
 $
  J(x;\pi,c)\le L(x;\varphi)$ for all $(\pi,c)\in\tilde{\calA}(x)$ and $\varphi\in\widetilde{\Theta}.
  $
\end{lemma}

\begin{proof}
From (\ref{ineqUI}) and (\ref{cY}), we have
\[
U_1(t,c_t)\le U_1(t,c_t^{x,\varphi})+\Y^{\varphi}(x)H_t^{\varphi}(c_t-c_t^{x,\varphi})
\]
and
\[
U_2(V_T^{x,\pi,c})\le U_2(G^{x,\varphi})+\Y^{\varphi}(x)H_T^{\varphi}(V_T^{x,\pi,c}-G^{x,\varphi}).
\]
Then, by (\ref{budget}) and the definition of $\Y^{\varphi},$ we have
\begin{align*}
 J(x;\pi,c) &\le L(x;\varphi)
 +\Y^{\varphi}(x)\cdot\Exp\left[\int_0^TH_t^{\varphi}(c_t-c_t^{x,\varphi})\,dt+H_T(V_T^{x,\pi,c}-G^{x,\varphi})\right]\\
 &\le L(x;\varphi)+\Y^{\varphi}(x)[x-\calX^{\varphi}(\Y^{\varphi}(x))]\\
 &= L(x;\varphi)
\end{align*}
and the desired result follows.
\end{proof}

Let $\tilde{\vartheta}(\cdot)$ denote the optimal value function of the minimization problem
\begin{equation}\label{dual}
\tilde{\vartheta}(x):=\inf_{\varphi\in\widetilde{\Theta}} L(x;\varphi).
\end{equation}
From Lemma \ref{ineqJL} we have $\vartheta(x)\le \tilde{\vartheta}(x).$ Our aim now is to find a sufficient condition for absence of duality gap and and existence of an optimal portfolio-consumption process $(\hat{\pi},\hat{c}),$ for a fixed initial wealth $x>0.$



For each $\varphi\in \tilde{\Theta},$ define the processes
\[
Y_t^{x,\varphi}:=\Exp\left[\Bigl.H_T^{\varphi} G^{x,\varphi}+\int_t^TH_s^{\varphi} c_s^{x,\varphi}\,ds\,\Bigr|\calF_t\right], \ \ t\in [0,T],
\]
and
\[
M_t^{x,\varphi}  := Y_t^{x,\varphi}+\int_0^tH_s^{\varphi}c_s^{x,\varphi}\,ds, \ \ t\in [0,T].
\]
Observe that $M_t^{x,\varphi}$ satisfies
\begin{equation}
M_t^{x,\varphi}=\Exp \left[ H_T^{\varphi} G^{x,\varphi}+\Bigl.\int_0^T H_s^{\varphi} c_s^{x,\varphi}\, ds\, \Bigr|  \calF_t \right], \ t \in [0,T].
\end{equation}
That is, the process $(M_t^{x,\varphi})_{t\in[0,T]}$ is a $\Fil$-martingale. Let $\beta^{x,\varphi}(t,y)$ denote the essentially unique martingale representation coefficient of $ M_t^{x,\varphi}$ with respect to the compensated measure $\tN(dy,dt),$  
\begin{equation}
{dM_t^{x,\varphi}} = {\int_{E}} {\beta^{x,\varphi}(t,y)}\,{\tN}(dy,dt).
\end{equation}
Notice that $Y_0^{x,\varphi} = \calX^{\varphi}(\Y^{\varphi}(x))= x$ and $Y_t^{x,\varphi} \geq 0 $ for all  $t \in \left[0,T\right].$ Moreover,
the pair $(Y^{x,\varphi},\beta^{x,\varphi})$ satisfies the linear backward SDE
\begin{equation}\label{bsdeY}
Y_t^{x,\varphi}=H_T^{\varphi} G^{x,\varphi}+\int_t^T H_s^{\varphi} c_s^{x,\varphi}\,ds-\int_t^T\!\int_E\beta^{x,\varphi}(s,y)\,{\tN}(dy,dt), \ \ t\in [0,T],
\end{equation}
with final condition $Y_T^{x,\varphi}=H_T^{\varphi} G^{x,\varphi}.$ The following is main result of this paper


\begin{theorem}\label{main}
For $x>0$ fixed, suppose there exist $\hat{\varphi} \in \tilde{\Theta}$ and a $\Fil$-predictable portfolio proportion process $\hat{\pi}$ with values in $K$ satisfying
\begin{equation}\label{condpizeta}
g(t,\hpi_t)-\hpi_t\zeta_t^{\hat{\varphi}}=\tg_K(t,\zeta^{\hat{\varphi}}_t), \ \ \mbox{a.s. for all } \ t\in [0,T],
\end{equation}
and
\begin{equation}\label{pif}
\hat{\pi}_tf(t,y)+1
=\frac{1}{\hat{\varphi}(t,y)}\left[\frac{\beta^{x,\hat{\varphi}}(t,y)}{Y_{t-}^{x,\hat{\varphi}}}+1\right], \ \ \mbox{a.s. \ for $\rho-$a.e.} \ (t,y)\in[0,T]\times E.
\end{equation}
Assume further (\ref{eqVnon-linear}) has a solution for $ (\hat{\pi} , \hat{c}),$ where $\hat{c}=c^{x,\hat{\varphi}}$. Then the following assertions hold

\begin{itemize}
   \item[(a)] $(\hat{\pi} , \hat{c})\in\tilde{\calA}(x)$ and  maximizes  (\ref{utilitymax}),

   \item[(b)] the wealth process $V^{x,\hat{\pi},\hat{c}}$ is a modification of the process $X_t^{x,\hat{\varphi}}:=Y_t^{x,\hat{\varphi}}/H_t^{\hat{\varphi}}$,

   \item[(c)] the optimal value function for the utility maximization (\ref{utilitymax}) satisfies $\vartheta=\mathcal{K}^{\hat{\varphi}}\circ\mathcal{Y}^{\hat{\varphi}}$ where
\[
\mathcal{K}^{\hat{\varphi}}(y)
:=\Exp\left[\int_0^TU_1(t,I_1(t,yH_t^{\hat{\varphi}}))\,dt+U_2(I_2(yH_T^{\hat{\varphi}}))\right], \ \ y>0.
\]
 \end{itemize}

\end{theorem}
\begin{proof}
We first prove part (b). Since $X_0^{x,\hat{\varphi}}=Y_0^{x,\hat{\varphi}}=x,$ it suffices to show that $X_t^{x,\hat{\varphi}}$ satisfies the wealth equation (\ref{eqVnon-linear}) for the pair $(\hat{\pi},\hat{c}).$ Recall that $H_t^{\hat{\varphi}}$ satisfies the linear stochastic equation
\begin{align*}
  dH_t^{\hat{\varphi}}&=H_{t-}^{\hat{\varphi}}\Bigl\{-[r_t+\tg_K(t,\zeta^{\hat{\varphi}}_t)]\,dt+\int_E(\hat{\varphi}(t,y)-1)\,\tN(dy,dt)\Bigr\}\\
  &=H_{t-}^{\hat{\varphi}}\Bigl\{-\Bigl[r_t+\tg_K(t,\zeta^{\hat{\varphi}}_t)+\lambda_t\int_E(\hat{\varphi}(t,y)-1)\,F_t(dy)\Bigr]\,dt+\int_E(\hat{\varphi}(t,y)-1)\,N(dy,dt)\Bigr\}
\end{align*}
Using Ito's formula for jump processes, the differential of $1/H_t^{\hat{\varphi}}$ is given by
\[
  d\biggl(\frac{1}{H_t^{\hat{\varphi}}}\biggr)
  =\frac{1}{H_{t-}^{\hat{\varphi}}}\Bigl\{\Bigl[r_t+\tg_K(t,\zeta^{\hat{\varphi}}_t)+\lambda_t\int_E(\hat{\varphi}(t,y)-1)\,F_t(dy)\Bigr]\,dt
  +\int_E\Bigl(\frac{1}{\hat{\varphi}(t,y)}-1\Bigr)\,N(dy,dt)\Bigr\}.
\]
From (\ref{bsdeY}), the differential of $Y_t$ is given by
\[
dY_t^{x,\hat{\varphi}}=-H_t^{\hat{\varphi}}c_t^{x,\hat{\varphi}}\,dt+\int_E\beta^{x,\hat{\varphi}}(t,y)\,\tN(dy,dt)
\]
Using the product rule for jump processes, we have
\begin{align*}
  d&\biggl(\frac{Y_t^{x,\hat{\varphi}}}{H_t^{\hat{\varphi}}}\biggr)
  =Y_{t-}^{x,\hat{\varphi}}\,d\Bigl(\frac{1}{H_t^{\hat{\varphi}}}\Bigr)
  +\frac{1}{H_{t-}^{\hat{\varphi}}}\,dY_t^{x,\hat{\varphi}}+\frac{1}{H_{t-}^{\hat{\varphi}}}\int_E\beta^{x,\hat{\varphi}}(t,y)\Bigl(\frac{1}{\hat{\varphi}(t,y)}-1\Bigr)\,N(dy,dt)\\
  &=\frac{Y_{t-}^{x,\hat{\varphi}}}{H_{t-}^{\hat{\varphi}}}\Bigl\{\Bigl[r_t+\tg_K(t,\zeta^{\hat{\varphi}}_t)+\lambda_t\int_E(\hat{\varphi}(t,y)-1)\,F_t(dy)\Bigr]\,dt
  +\int_E\Bigl(\frac{1}{\hat{\varphi}(t,y)}-1\Bigr)\,N(dy,dt)\Bigr\}\\
  &\phantom{==}-c_t^{x,\hat{\varphi}}\,dt+\frac{1}{H_{t-}^{\hat{\varphi}}}\Bigl\{\int_E\beta^{x,\hat{\varphi}}(t,y)\,\tN(dy,dt)
  +\int_E\beta^{x,\hat{\varphi}}(t,y)\Bigl(\frac{1}{\hat{\varphi}(t,y)}-1\Bigr)\,N(dy,dt)\Bigr\}
\end{align*}
We multiply and divide the last bracket by $Y_{t-}^{x,\hat{\varphi}}$ and use $\tN(dy,dt)=N(dy,dt)-\lambda_t\,F_t(dy)\,dt$ to obtain
\begin{align*}
d\biggl(\frac{Y_t^{x,\hat{\varphi}}}{H_t^{\hat{\varphi}}}\biggr)
  &=\frac{Y_{t-}^{x,\hat{\varphi}}}{H_{t-}^{\hat{\varphi}}}\left\{[r_t+\tg_K(t,\zeta^{\hat{\varphi}}_t)
  +\lambda_t\int_E\Bigl(\varphi(t,y)-1-\frac{\beta^{x,\hat{\varphi}}(t,y)}{Y_{t-}^{x,\hat{\varphi}}}\Bigr)\,F_t(dy)\,dt\right.\\
&\phantom{==}\left.+\int_E\Bigl(\frac{1}{\hat{\varphi}(t,y)}-1+\frac{\beta^{x,\hat{\varphi}}(t,y)}{\hat{\varphi}(t,y)Y_{t-}^{x,\hat{\varphi}}}\Bigr)\,N(dy,dt)\right\}
-c_t{x,\hat{\varphi}}\,dt
\end{align*}
From (\ref{pif}), for the integrand in the stochastic integral, we have
\[
\frac{1}{\hat{\varphi}(t,y)}-1+\frac{\beta^{x,\hat{\varphi}}(t,y)}{\hat{\varphi}(t,y)Y_{t-}^{x,\hat{\varphi}}}
  =\hpi_tf(t,y)
\]
and (\ref{pif}) in conjunction with (\ref{condpizeta}) yields
\begin{align*}
\tg_K(t,\zeta^{\hat{\varphi}}_t)&+\lambda_t\int_E\Bigl(\hat{\varphi}(t,y)-1-\frac{\beta^{x,\hat{\varphi}}(t,y)}{Y_{t-}^{x,\hat{\varphi}}}\Bigr)\,F_t(dy)\\
&=\tg_K(t,\zeta^{\hat{\varphi}}_t)+\lambda_t\int_E(\hat{\varphi}(t,y)-\hat{\varphi}(t,y)(\hpi_t f(t,y)+1))\,F_t(dy)\\
&=\tg_K(t,\zeta^{\hat{\varphi}}_t)+\lambda_t\int_E-\hpi_t\hat{\varphi}(t,y)f(t,y)\,F_t(dy)\\
&=g(t,\hat{\pi}_t)+\hpi_t(\mu_t-r_t)
\end{align*}
and part (b) follows. This in turn implies that $V_T^{x,\hpi,\hat{c}}=Y_T^{x,\hat{\varphi}}/H_T^{\hat{\varphi}}=G^{x,\hat{\varphi}},$ a.s. In particular, we get
\begin{equation}\label{eqJL}
J(x;\hat{\pi},\hat{c})=L(x;\hat{\varphi})
\end{equation}
and part (a) follows from Lemma \ref{ineqJL}. Part (c) follows easily since $\vartheta(x)=\tilde{\vartheta}(x)=L(x;\hat{\varphi}).$

\end{proof}


\begin{remark}
From (\ref{condpizeta}) and the definition of $\tg_K,$ the  portfolio process $\hat{\pi}$ satisfies
\[
\hat{\pi}_t=\arg\max_{\pi\in K}\set{g(t,\pi)-\pi\Bigl[r_t-\mu_t-\lambda_t\int_E f(t,y)\hat{\varphi}(t,y)\,F_t(dy)\Bigr]}
\]
or, equivalently,
\[
r_t-\mu_t-\lambda_t\int_E f(t,y)\hat{\varphi}(t,y)\,F_t(dy)
=\arg\min_{\zeta\in\R}\left(\tg_K(t,\zeta)+\hpi_t\zeta\right).
\]
\end{remark}

\section{Examples}
\subsection{Logarithmic utility}
We illustrate the main result first by considering logarithmic utility functions $U_1(t,x)=U_2(x)=\ln x.$
\begin{lemma}\label{betalog}
For all $\varphi\in\tilde{\Theta}$ and $x>0$ we have $\beta^{x,\varphi}(t,y)=0,$ a.s. for $\rho$-a.e. $(t,y)\in [0,T]\times E.$
\end{lemma}
\begin{proof}
In this case, we have $I_1(t,y)=I_2(y)=1/y$ and $\calX^{\varphi}(y)=(T+1)/y,$ for $y\in(0,\infty).$ Then, $\Y^{\varphi}(x)=(T+1)/x$ for $x>0$ and
\begin{align}
  c_t^{x,\varphi} &=\frac{x}{(T+1)H_t^{\varphi}}, \ \ t\in [0,T],\label{clog}\\
  G^{x,\varphi}&=\frac{x}{(T+1)H_T^{\varphi}}.\notag
\end{align}
Hence, $M_t^{x,\varphi}=x$ for all $t\in [0,T],$ and the desired result follows.
\end{proof}

\begin{theorem}\label{thmlog}
Let $x$ be fixed. Suppose there exists a $\Fil-$predictable portfolio process $(\hpi_t)_{t\in [0,T]}$ with values in $K$ satisfying
\begin{enumerate}
  \item[(i)] $1+\hpi_tf(t,y)>0$ a.s. for $\rho$-a.e. $(t,y)\in [0,T]\times E$
  \item[(ii)] The process
\begin{equation}\label{hnulog}
\hat\zeta_t:=r_t-\mu_t-\lambda_t\int_E \frac{f(t,y)}{1+\hpi_tf(t,y)}\,F_t(dy), \ \ t\in [0,T],
\end{equation}
belongs to $\mathcal{D}$ and satisfies
\begin{equation}\label{hpizetalog}
g(t,\hpi_t)-\hpi_t\hat\zeta_t=\tg_K(t,\hat\zeta_t), \ \ \mbox{a.s. for all} \ t\in [0,T].
\end{equation}
\end{enumerate}
Then the pair $(\hpi,\hat{c})$ is optimal, where $\hat{c}=(\hat{c}_t)_{t\in [0,T]}$ is the consumption process defined by
\[
\hat{c}_t:=\frac{x}{(1+T)}V_t^{1,\hpi,0}, \ \ t\in [0,T],
\]
and $V^{1,\hat{\pi},0}=(V_t^{1,\hat{\pi},0})_{t\in [0,T]}$ is the wealth process with initial wealth $1$ and portfolio-consumption pair $(\hat{\pi},0).$ Moreover, the optimal wealth process $V^{x,\hat{\pi},\hat{c}}$ satisfies
 \[
 V^{x,\hat{\pi},\hat{c}}_t=V_t^{x,\hat{\pi},0}-t\hat{c}_t=V_t^{x,\hat{\pi},0}\Bigl(1-\frac{t}{T+1}\Bigr), \ \ t\in [0,T].
 \]
\end{theorem}

\begin{proof}
Define $\hat\varphi(t,y):=1/[1+\hat{\pi}_tf(t,y)].$ Then $\hat\varphi\in\tilde\Theta$ and $\zeta^{\hat\varphi}=\hat\zeta.$ By Lemma \ref{betalog}, $\hat\varphi$ and $\hat\pi$ satisfy the assumptions of Theorem \ref{main}. Then the pair $(\hat\pi,c^{x,\hat{\varphi}})$ is optimal.

Using again (\ref{hpizetalog}), we see that the differential of $(H_t^{\hat{\varphi}})^{-1}$ satisfies
\begin{align*}
  d\Bigl(\frac{1}{H^{\hat{\varphi}}_t}\Bigr)
  &=\frac{1}{H^{\hat{\varphi}}_{t-}}\Bigl\{\Bigl[r_t+\tg_K(t,\hat\zeta_t)
  +\lambda_{t}\int_E(\hat{\varphi}(t,y)-1)\,F_{t}(dy)\Bigr]\,dt+\int_E\Bigl(\frac{1}{\hat{\varphi}(t,y)}-1\Bigr)\,N(dy,dt)\Bigr\}\\
  &=\frac{1}{H^{\hat{\varphi}}_{t-}}\left\{\Bigl[r_t+g(t,\hat\pi_t)\Bigr]\,dt+\hat\pi_{t}\left[\bigl(\mu_{t}-r_{t}\bigr)\,dt+\int_{E} f(t,y)\,N(dy,dt)\right]\right\}.
\end{align*}
Hence, the process $(H_t^{\hat{\varphi}})^{-1}$ is a modification of $V_t^{1,\hat{\pi},0}.$ In view of (\ref{clog}), we conclude  $\hat{c}=c^{x,\hat{\varphi}}$ and the first assertion follows. The second assertion follows from (\ref{xiV}). 
\end{proof}




\subsection{Regime-switching model with Markov-modulated jump-size distributions}
Here we consider the pure-jump model with Markov-modulated jump-size distributions from L\'opez and Ratanov \cite{lopezrat} (see also L\'opez and Serrano \cite{lopezserrano}) and logarithmic utility.

Let $\set{(\tau_n,Y_{\eps_n,n})}_{n\geq 1}$ be the marked point process from Example \ref{MMjumps} with $E=\R.$ The random times $\set{\tau_n}_{n\geq 1}$ are defined as the jump times of a the two-state continuous-time Markov chain $\eps(\cdot)=\{\eps(t)\}_{t\geq0}$ with intensity matrix
\begin{equation*}
Q=
\begin{pmatrix}
-\bm{\lambda}_0 & \ \ \bm{\lambda}_0\\
\bm{\ \ \lambda}_1 & -\bm{\lambda}_1
\end{pmatrix}.
\end{equation*}
For each $n\geq 1,$ $\eps_n:=\eps(\tau_n-)$ is the state right before the $n-$th jump of $\eps(\cdot)$ and the mark $Y_{\eps_n,n}$ is a random variable with distribution $\bm{F}_{\eps_n}(dy).$

For each $i=0,1,$ let ${r}_i>0$ and $\mu_i$ denote the continuously compounded interest rate and stock appreciation rate in the regime $i$ respectively. Let $B_t$ denote the  default-free money-market account with Markov-modulated force of interest $\{{r}_{\eps(t)}\}_{t\in[0,T]},$ that is
\[
B_t=\exp\left(\int_0^t{r}_{\eps(s)}ds\right), \ \ t\in[0,T].
\]
The risky asset or stock $S_t$ follows the exponential model $S_t=S_0\exp(\tilde{L}_t)$ with $S_0>0$ and
\[
\tilde{L}_t=\int_0^t\mu_{\eps(s)}\,ds+\sum_{n=1}^{N_t(E)}Y_{\eps_n,n}, \ \ t\in [0,T]
\]
Observe that $S_t$ satisfies the linear equation (\ref{eqS}) with $f(t,y)=e^y-1.$


We assume that for each regime $i=0,1$ there exists a margin payment function $g_i(\pi)$ with portfolio constraint set $K.$ For instance, in the case of different interest rates for borrowing and lending, and prohibition of short-selling, it is given by
\[
g_i(\pi)=-({R}_i-{r}_i)(\pi-1)^+, \ \ \pi\in K=[0,\infty)
\]
where ${R}_i$ denotes the borrowing rate in regime $i$, which is assumed greater than the lending rate ${r}_i.$ In the case of short selling with cash collateral and negative rebate rates, and prohibition of borrowing from money account,
\[
g_i(\pi)=(r_i-r_i^L)\pi^{-}, \ \ \pi\in K=(-\infty,1]
\]
where $r_i^L$ is the stock loan fee in regime $i.$ Finally, for each $i=0,1$ we define
\[
\tg_i(\zeta):=\sup_{\pi\in K}[g_i(\pi)-\pi\zeta], \ \ \zeta\in\R
\]
and $\N_i:=\set{\zeta\in\R:\tg_i(\zeta)<+\infty}.$ Using Theorem \ref{thmlog} and the results in Section 4 from \cite{lopezserrano}, we obtain the following
\begin{cor}
Let $U_1(t,x)=U_2(x)=\ln x.$ Suppose for each $i=0,1$ there exists $\bar\pi_i\in K_i$ such that $1+\bar\pi_i (e^y-1)>0$ for all $y\in\supp\bm{F}_i.$ Suppose further the following conditions hold
\begin{enumerate}
  \item[(i)] $\displaystyle{\bar\eta_i:=\int_{\R}\ln(1+\bar\pi_i(e^y-1))\,\bm{F}_i(dy)<+\infty}$

  \smallskip\item[(ii)] $\displaystyle{\bar\zeta_i:=r_i-\mu_i-\bm{\lambda}_i\int_\R \frac{e^y-1}{1+\bar\pi_i(e^y-1)}\,\bm{F}_i(dy)\in\N_i}$

\smallskip\item[(iii)] $g_i(\bar\pi_i)-\bar\pi_i\bar\zeta_i=\tg_i(\bar\zeta_i).$

\end{enumerate}
Let $\vartheta_i(x)$ denote the optimal value for the initial wealth $x>0$ and initial regime $\eps(0)=i.$ Then, we have
\begin{align*}
\vartheta_0(x)&=(T+1)\ln x-(T+1)\ln (T+1)\\
&\phantom{=}-\frac{1}{2\bm{\lambda}}\left\{(\bm{\lambda}_1\bar{d}_0+\bm{\lambda}_0\bar{d}_1)\Bigl(T+\frac{T^2}{2}\Bigr)
+\frac{\bm{\lambda}_0(\bar{d}_0-\bar{d}_1)}{2\bm{\lambda}}\Bigl[T+\bigl(1-e^{-2\bm{\lambda} T}\bigr)\Bigl(1+\frac{1}{2\bm{\lambda}}\Bigr)\Bigr]\right\}
\end{align*}
and
\begin{align*}
\vartheta_1(x)&=(T+1)\ln x-(T+1)\ln (T+1)\\
&\phantom{=}-\frac{1}{2\bm{\lambda}}\left\{(\bm{\lambda}_1\bar{d}_0+\bm{\lambda}_0\bar{d}_1)\Bigl(T+\frac{T^2}{2}\Bigr)
-\frac{\bm{\lambda}_1(\bar{d}_0-\bar{d}_1)}{2\bm{\lambda}}\Bigl[T+\bigl(1-e^{-2\bm{\lambda} T}\bigr)\Bigl(1+\frac{1}{2\bm{\lambda}}\Bigr)\Bigr]\right\}
\end{align*}
where
\begin{equation*}
2\bm{\lambda}:=\bm{\lambda}_0+\bm{\lambda}_1 \ \ \mbox{ and } \ \ \bar{d}_i:=\bar{\pi}_i\mu_i+(1-\bar\pi_i)r_i+\bm{\lambda}_i\bar{\eta}_i, \ \ i=0,1.
\end{equation*}
\end{cor}

\subsection{Power utility}

We now consider CRRA (fractional) power utility functions of the form $U_1(t,x)=U_2(x)=\frac{x^\gamma}{\gamma}$ with $\gamma\in(0,1)$ fixed. We suppose that all coefficients in the model $r_t,\mu_t,f(t,y),$ the $\Fil$-local characteristics $(\lambda_t,F_t(dy))$ and the margin payment function $g(t,\pi)$ are non-random.

\begin{lemma}
For all $x>0$ and $\varphi\in\tilde\Theta$ deterministic, we have
\begin{equation}\label{betapower}
\beta^{x,\varphi}(t,y)= Y_{t-}^{x,\varphi}\bigl(\varphi(t,y)^{\frac{\gamma}{\gamma-1}}-1\bigr), \ \ \rho\mbox{-a.e.} \  (t,y)\in [0,T]\times E.
\end{equation}
\end{lemma}
\begin{proof}
Notice that $(H_t^{\varphi})^{\frac{\gamma}{\gamma-1}}=h_t\tH_t$ where $h_t$ is the deterministic function
\begin{align*}
h_t=\exp&\left(\int_0^t\Bigl\{\frac{-\gamma}{\gamma-1}[r_s+\tg_K(t,\zeta^{\varphi}_s)]\Bigr.\right.\\
&\phantom{AA}\left.\Bigl.+\lambda_s\int_E\Bigl[\varphi(s,y)^{\frac{\gamma}{\gamma-1}}-1+\frac{\gamma}{\gamma-1}(1-\varphi(s,y))\Bigr]F_s(dy)\Bigr\}\,ds\right)
\end{align*}
and
$\tH_t$ is the $\Fil$-martingale
\[
d\tH_t=\tH_t\int_E\bigl(\varphi(t,y)^{\frac{\gamma}{1-\gamma}}-1\bigr)\,\tilde N(dy,dt), \ \ \tH_0=1.
\]
Then
\[
\calX^\varphi(y)=y^{\frac{1}{\gamma-1}}\Exp\left[\int_0^T (H_t^\varphi)^{\frac{\gamma}{\gamma-1}}\,dt+(H_T^\varphi)^{\frac{\gamma}{\gamma-1}}\right]=\kappa y^{\frac{1}{\gamma-1}}
\]
with $\kappa:=h_T+\int_0^Th_t\,dt.$ It follows that
\[
\Y^{\varphi}(x)=\Bigl(\frac{x}{\kappa}\Bigr)^{\gamma-1}, \quad x>0
\]
and
\begin{align}
  c_t^{x,\varphi} &=\frac{x}{\kappa}(H_t^\varphi)^{\frac{1}{\gamma-1}}, \ \ t\in [0,T],\label{cpower}\\
  G^{x,\varphi}&=\frac{x}{\kappa}(H_T^\varphi)^{\frac{1}{\gamma-1}}.\notag
\end{align}
Hence
\begin{align*}
Y_t^{x,\varphi}&=\frac{x}{\kappa}\Exp\left[(H_T^\varphi)^{\frac{\gamma}{\gamma-1}}+\int_t^T(H_s^\varphi)^{\frac{\gamma}{\gamma-1}}\,ds\,\Bigl|\,\calF_t\Bigr.\right]\\
  &=\tH_t\frac{x}{\kappa}\Bigl[h_T+\int_t^T h_s\,ds\Bigr]
\end{align*}
and
\[
M_t^{x,\varphi}=\frac{x}{\kappa}\left\{\tH_t\Bigl[h_T+\int_t^T h_s\,ds\Bigr]+\int_0^t\tH_sh_s\,ds\right\}, \ \ t\in [0,T].
\]
The differential of $M_t^{x,\varphi}$ satisfies
\begin{align*}
dM_t^{x,\varphi}&=\frac{x}{\kappa}\Bigl[h_T\,d\tH_t+\Bigl(\int_t^T h_s\,ds\Bigr)\,d\tH_t\Bigr]\\
&=Y_{t-}^{x,\varphi} \int_E\bigl(\varphi(t,y)^{\frac{\gamma}{1-\gamma}}-1\bigr)\,\tilde N(dy,dt)
\end{align*}
and (\ref{betapower}) follows.
\end{proof}

For simplicity suppose consumption is not allowed i.e. $c_t=0$ for all $t\in [0,T].$
\begin{theorem}\label{thmpower}
Let $\hpi=(\hpi_t)_{t\in [0,T]}$ be a (deterministic) portfolio process with values in $K$ satisfying
\begin{enumerate}
  \item[(i)] $1+\hpi_tf(t,y)>0$ for all $(t,y)\in [0,T]\times E,$
  \item[(ii)] The map
\[
\hat\zeta_t:=r_t-\mu_t-\lambda_t\int_E \frac{f(t,y)}{[1+\hpi_tf(t,y)]^{1-\gamma}}\,F_t(dy), \ \ t\in [0,T],
\]
belongs to $\mathcal{D}$ and satisfies
\[
g(t,\hpi_t)-\hpi_t\hat\zeta_t=\tg_K(t,\hat\zeta_t), \ \ \mbox{a.s. for all} \ t\in [0,T].
\]
\end{enumerate}
Then the portfolio process $\hpi$ is optimal.
\end{theorem}
\begin{proof}
Define $\hat\varphi(t,y):=1/[1+\hat\pi_t f(t,y)]^{1-\gamma}.$ Then $\hat\varphi\in\tilde\Theta$ and $\zeta^{\hat\varphi}=\hat\zeta.$ From (\ref{betapower}) it follows that $\hat\varphi$ and $\hat\pi$ satisfy the assumptions of Theorem \ref{main}, and the desired result follows.
\end{proof}

\begin{example}\label{exdiffrates3}
Consider the margin payment function of Example \ref{exdiffrates}
\[
g(t,\pi_t)=-(R_t-r_t)(\pi_t-1)^+
\]
which models differential interest rates, under the portfolio constraint of prohibition of short-selling of the stock. As seen in Example \ref{exdiffrates2}, the effective domain of $\tg_K$ is  $\N_t=[-(R_t-r_t),+\infty).$

Assume further that $1+\pi f(t,y)>0$ for all $(t,y)\in [0,T]\times E$ and $\pi\geq 0,$ and the map
\begin{equation}\label{hdrexpl}
h_t(\pi):=\mu_t+\lambda_t\int_E \frac{f(t,y)}{[1+\pi f(t,y)]^{1-\gamma}}\,F_t(dy), \ \ \pi\geq 0
\end{equation}
is well-defined for each $t\in [0,T]$ and $\gamma\in [0,1).$ The case $\gamma=0$ corresponds to logarithmic utility, in which case we allow $r_t,R_t,\mu_t,f(t,y)$ and the $\Fil$-local characteristics $(\lambda_t,F_t(dy))$ to be $\Fil$-predictable processes. It follows that
\[
q_t(\pi):=r_t-h_t(\pi)=r_t-\mu_t-\lambda_t\int_E \frac{f(t,y)}{[1+\pi f(t,y)]^{1-\gamma}}\,F_t(dy),
\]
belongs to $\mathcal{N}_t$ iff $h_t(\pi)\le R_t.$ Using (\ref{gkdiffrates}), condition $g(t,\pi_t)-\pi_t q_t(\pi_t)=\tg_K(t,q_t(\pi_t))$ reads
\begin{equation}\label{pidiffrates}
\left[r_t-h_t(\pi_t)\right]^-+\pi_t\left[r_t-h_t(\pi_t)\right]+(R_t-r_t)(\pi_t-1)^+=0.
\end{equation}
Now, observe that $h_t(\cdot)$ is strictly decreasing for each $t\in [0,T]$ since
\[
h_t'(\pi):=\lambda_t(\gamma-1)\int_E\frac{f(t,y)^2}{[1+\pi f(t,y)]^{2-\gamma}}\,F_t(dy) < 0, \ \ \pi\geq 0.
\]
\begin{figure}[t!]
\centering
\includegraphics[scale=0.6]{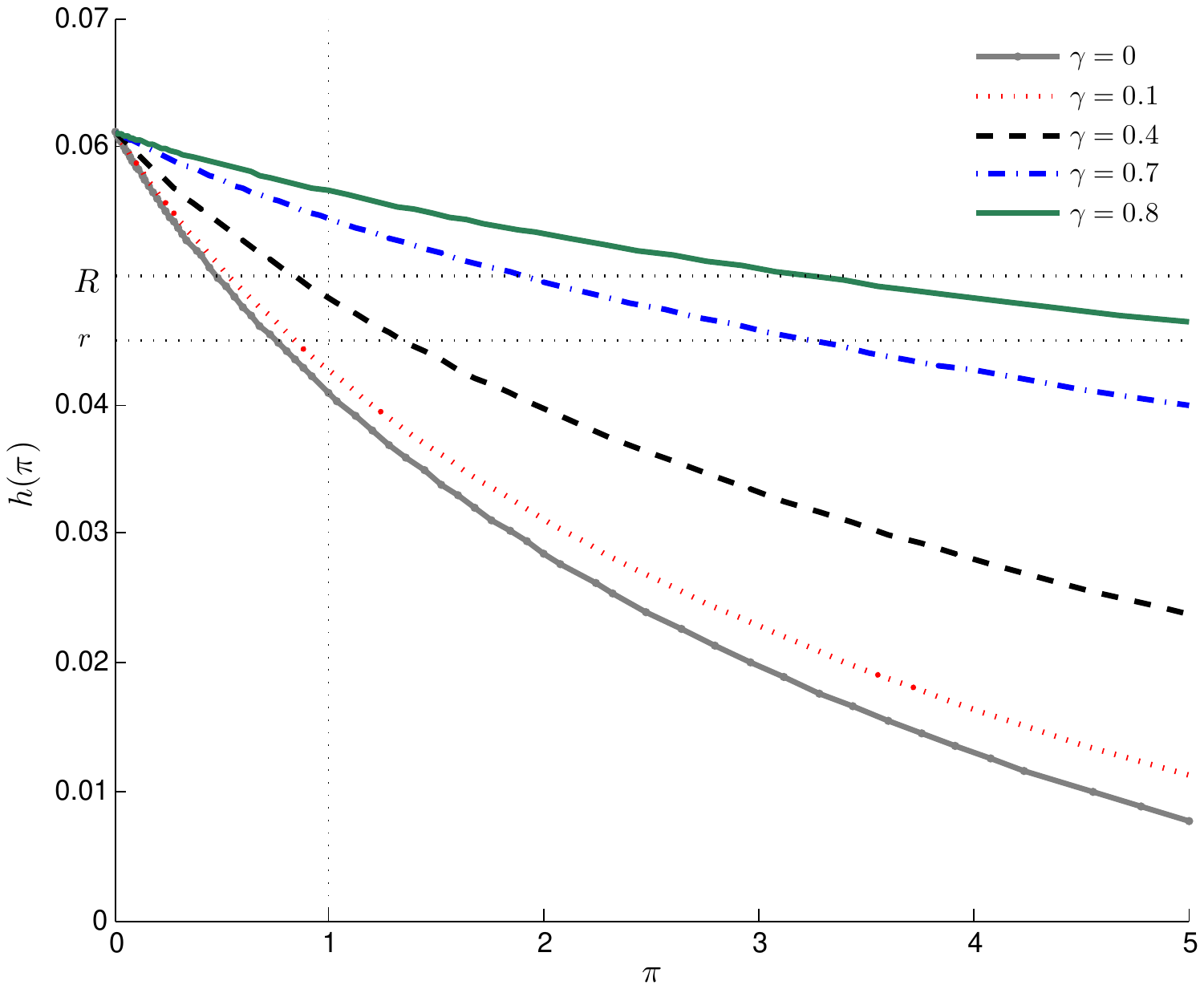}
\caption{Plot of $h$ given by \eqref{hdrexpl} for different values of $\gamma$ with $\mu=-0.05$, $\lambda=1$, $r=0.045$, $R=0.05$, $f(y)=e^y-1$ and $F(dy)=10e^{-10y}1_{y\geq0}\,dy$.}\label{dr1}
\end{figure}
The range of $h_t$ is the interval $(\mu_t,h_t(0)].$ If $R_t>\mu_t$ also holds, it can be easily checked that the following portfolio weight
\begin{equation}\label{pioptdr}
\hat\pi_t=\left\{
  \begin{array}{ll}
    0, & \ h_t(0)<r_t\\

    h_t^{-1}(r_t), &  \  h_t(1)\leq r_t< h_t(0)\\

    1, &  \ r_t<h_t(1)\leq R_t \\

    h_t^{-1}(R_t), &  \  R_t<h_t(1)
  \end{array}
\right.
\end{equation}
satisfies $h_t(\hat\pi_t)\le R_t$ and (\ref{pidiffrates}). Hence, it is optimal. Observe that under the condition $h_t(0)<r_t$ the optimal portfolio $\hat{\pi}_t=0$ does not depend on the risk aversion exponent $\gamma$.

Notice also that $h_t(0)$ and $h_t(1)$ can be given explicitly in terms of the moment generating function $m_t(\cdot)$ of the distribution $F_t(dy).$ Indeed, $h_t(0)=\mu_t+\lambda_t[m_t(1)-1]$ and $h_t(1)=\mu_t+\lambda_t[m_t(\gamma)-m_t(\gamma-1)].$

Figure \ref{dr1} contains plots of $h_t(\pi)$ for a given $t\in [0,T]$ and different values of $\gamma.$  Figure \ref{piopt} below shows the behaviour of the optimal portfolio $\hat{\pi}_t$ as a function of $\gamma.$ Observe that we can obtain each one of the three last cases in \eqref{pioptdr} by taking different values of $\gamma.$
\begin{figure}[t!]
\centering
\includegraphics[scale=0.6]{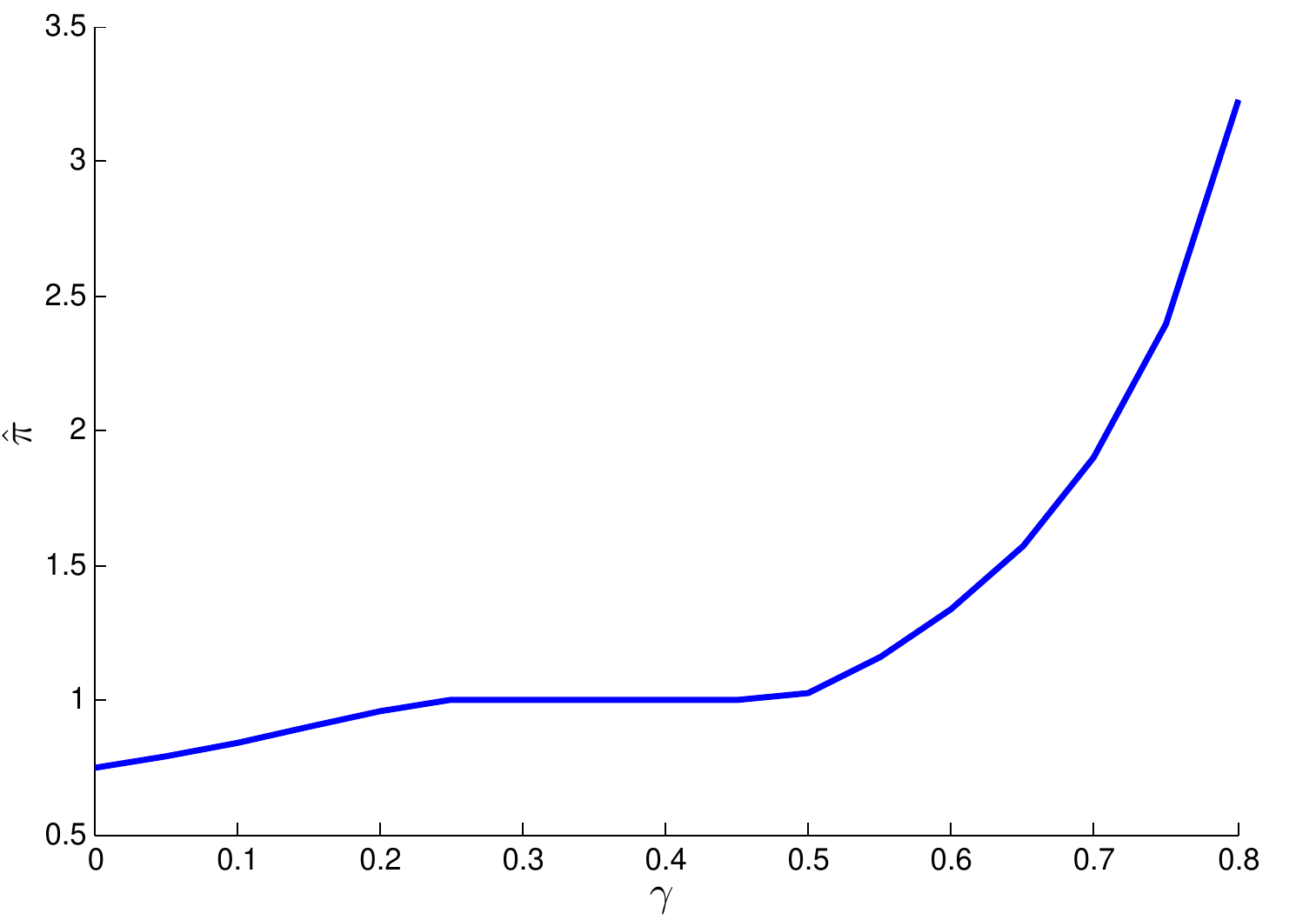}
\caption{Plot of $\hat\pi$ given by \eqref{pioptdr} as a function of $\gamma$.}\label{piopt}
\end{figure}
\end{example}


\begin{example}\label{margin3}
Finally, consider the margin payment function of Example \ref{margin}
\[
g(t,\pi_t)=(r_t-r_t^L)\pi_t^-
\]
with $K=(-\infty,1],$ which models short-selling with negative rebate rates and prohibition of borrowing from money account. The effective domain of $\tg_K$ is $\N_t=(-\infty,r_t^L-r_t].$

Hence, $q_t(\pi_t)\in\N_t$ iff $h_t(\pi_t)\geq 2r_t-r_t^L$ and condition $g(t,\pi_t)-\pi_t q_t(\pi_t)=\tg_K(t,q_t(\pi_t))$ reads
\begin{equation}\label{pimargin}
(r_t-r_t^L)\pi_t^- +\pi_t[h_t(\pi_t)-r_t]=[h_t(\pi_t)-r_t]^+.
\end{equation}
If $\mu_t>2r_t-r_t^L$ also holds, it can be easily checked that the following portfolio weight
\begin{equation}\label{pioptshort}
\hat\pi_t=\left\{
  \begin{array}{ll}
    h_t^{-1}(2r_t-r_t^L), & \ h_t(0)<2r_t-r_t^L\\

    0, &  \  2r_t-r_t^L\le h_t(0)<r_t\\

    h_t^{-1}(r_t), &  \ h_t(1)<r_t\le h_t(0)\\

    1, &  \  r_t\le h_t(1)
  \end{array}
\right.
\end{equation}
satisfies $h_t(\pi_t)\geq 2r_t-r_t^L$ and (\ref{pimargin}). Therefore, it is optimal. Observe that, under condition $2r_t-r_t^L\le h_t(0)<r_t,$ the optimal portfolio $\hat{\pi}_t=0$ does not depend on the risk aversion exponent $\gamma$.

Figures \ref{short} and \ref{short2}  contain plots of $h_t(\pi)$ for different values of $\gamma$ and of the optimal portfolio $\hat{\pi}_t$ as a function of $\gamma.$
\begin{figure}[t!]
\centering
\includegraphics[scale=0.6]{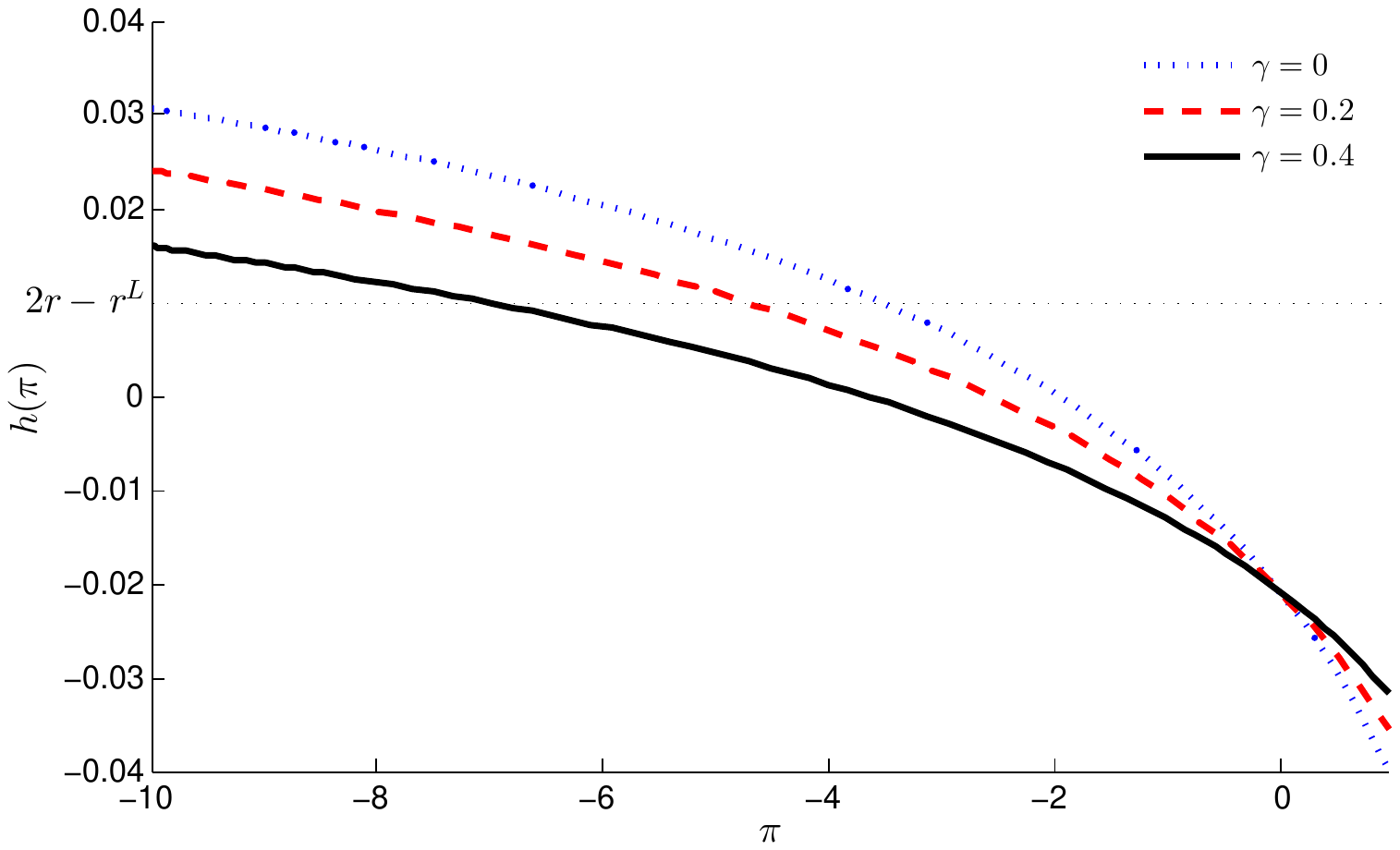}
\caption{Plot of $h$ given by \eqref{hdrexpl} for different values of $\gamma$ with $\mu=0.07$, $\lambda=1$, $r=0.03$, $r^L=0.05$, $f(y)=e^y-1$ and $F(dy)=10e^{-10y}1_{y\leq0}\,dy$.}\label{short}
\end{figure}

\begin{figure}[t!]
\centering
\includegraphics[scale=0.6]{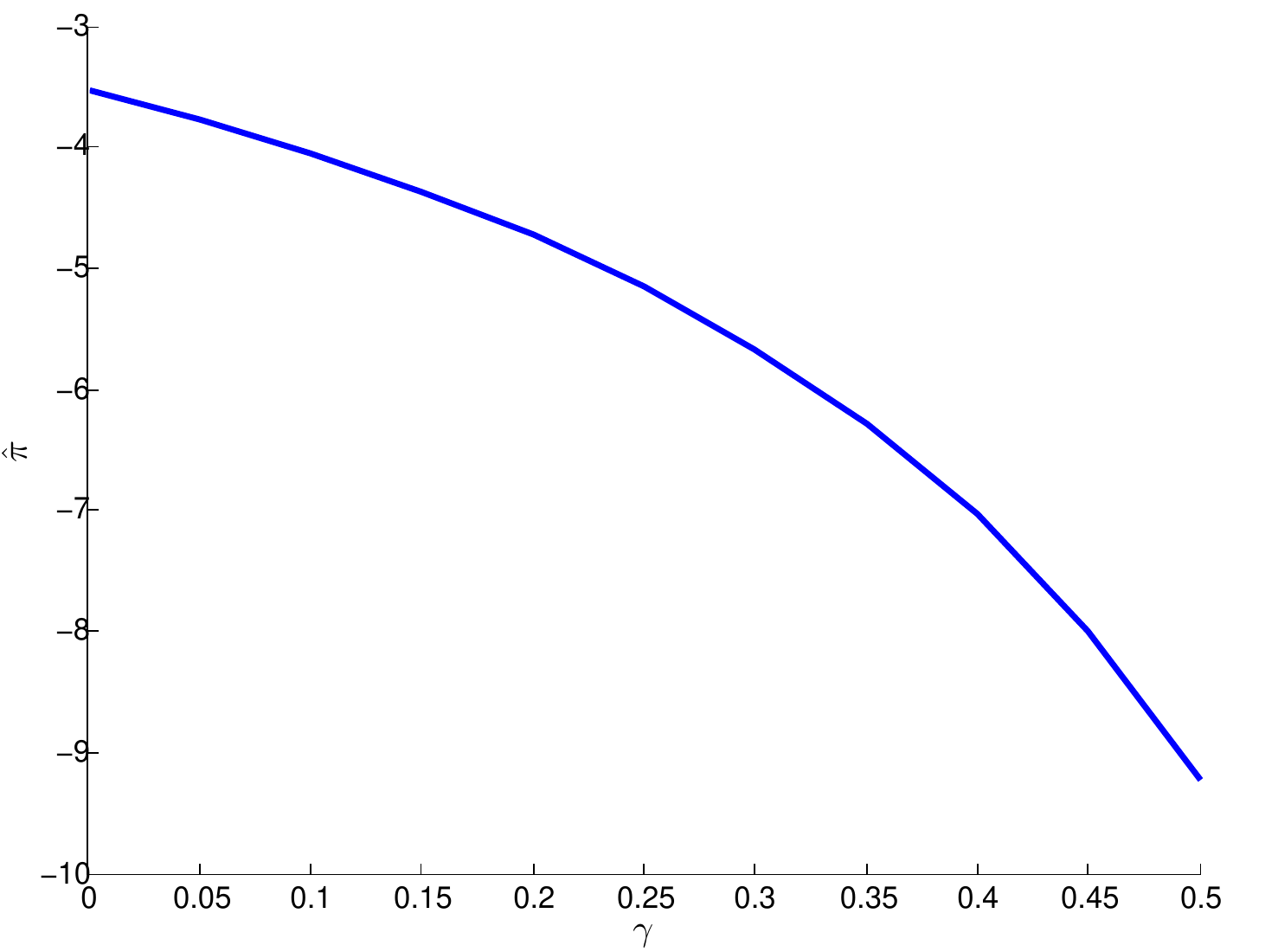}
\caption{Plot of $\hat\pi$ given by \eqref{pioptshort} as a function of $\gamma$.}\label{short2}
\end{figure}

\end{example}

\bibliography{biblioMPP}
\end{document}